\documentclass[11pt]{amsart}

\usepackage{amssymb,latexsym,amsmath,extarrows, mathrsfs, amsthm,color,amsrefs}
\usepackage{graphicx}

\usepackage[margin=1in, centering]{geometry}
\usepackage[bookmarksnumbered, colorlinks, plainpages]{hyperref}
\usepackage{hyperref} 
\hypersetup{
    colorlinks=true,       
    linkcolor=blue,          
    citecolor=magenta,        
    filecolor=magenta,      
    urlcolor=cyan           
}

\usepackage{mathtools}
\mathtoolsset{showonlyrefs,showmanualtags}

\newtheorem{theorem}{Theorem}[section]

\newtheorem{lemma}[theorem]{Lemma}

\newtheorem{remark}[theorem]{Remark}

\newtheorem*{definition}{Definition}
\newtheorem{corollary}[theorem]{Corollary}

\newtheorem{fact}[theorem]{Fact}

\newcommand{\Z}{\mathbb Z}
\newcommand{\R}{\mathbb R}
\newcommand{\C}{\mathbb C}
\newcommand{\T}{\mathbb T}

\newcommand{\Q}{\mathbb Q}

\definecolor{deepgreen}{cmyk}{1,0,1,0.5}

\title[Avila's acceleration, zeros, localization]{Avila's acceleration via zeros of determinants, and applications to Schr\"odinger cocycles}

\author[R.\ Han]{Rui Han}
\address{Department of Mathematics \\ Louisiana State University  \\  Baton Rouge, LA 70803, USA}
\email{rhan@lsu.edu}

\author[W.\ Schlag]{Wilhelm Schlag}
\address{Department of Mathematics \\ Yale University \\ New Haven, CT 06511, USA}
\email{wilhelm.schlag@yale.edu}

\thanks{
R.\ Han is partially supported by NSF DMS-2143369. 
W.\ Schlag is partially supported by NSF grant DMS-1902691.
We would like to thank Michael Goldstein for his interest and encouragement,
and for suggesting Remark~\ref{rm:DC}}

\begin{document}
\begin{abstract}
In this paper we give a characterization of Avila's quantized acceleration of the Lyapunov exponent via the number of zeros of the Dirichlet determinants in finite volume.
As applications, we prove $\beta$-H\"older continuity of the integrated density of states for supercritical quasi-periodic Schr\"odinger operators restricted to the $\ell$-th stratum, for any $\beta<(2(\ell-1))^{-1}$ and $\ell\ge2$.  We establish Anderson localization for all Diophantine frequencies for the operator with even analytic potential function on the first supercritical stratum, which has positive measure if it is nonempty.
\end{abstract}

\thanks{}

\maketitle

\section{Introduction}
This paper studies the one-dimensional quasi-periodic Schr\"odinger operator:
\begin{align}\label{def:H_operator}
(H_{\alpha,\theta}\phi)_n=\phi_{n+1}+\phi_{n-1}+f(\theta+n \alpha)\phi_n,
\end{align}
with real-valued analytic potential $f(\theta)\in C_{\eta}^{\omega}(\T)$, $\eta>0$.
Here $\alpha$ is called the frequency and $\theta$ is called the phase, and $\T:=\R/\Z$ is the one-dimensional torus, 
and we shall use $\|\cdot \|_{\T}$ to denote the torus norm.

We study the operator in the positive Lyapunov exponent regime.
The goal of this paper is to build a bridge between the large deviation estimates and avalanche principle developed in a series of papers by Bourgain-Goldstein \cite{BG}, Goldstein-Schlag \cites{GS1,GS2} on the one hand, with Avila's quantized acceleration of the Lyapunov exponent \cite{Global} on the other hand.
In particular, we obtain a sharp characterization of the acceleration in terms of the number of zeros of the Dirichlet determinants, see Theorem \ref{thm:acc=zeros}. 
As applications, we improve the H\"older exponent of the integrated density of states in \cite{GS2} (see Theorem \ref{thm:IDS}), and we prove Anderson localization of the operator on the first supercritical stratum, on which the acceleration equals~$1$, and  for all Diophantine frequencies (see Theorem \ref{thm:loc}).

To be more specific, let 
\begin{align*}
D_n(\theta,E):=\det(H_{\alpha,\theta}-E)|_{[0,n-1]}.
\end{align*}
Since $f$ is analytic, we can write $D_n(\theta, E)$ as a function of $e^{2\pi i \theta}$. Replacing $e^{2\pi i \theta}$ with a complex variable $z$, we shall also denote $D_n(\theta,E)$ by $D_n(z,E)$.
We assume throughout that $\alpha$ satisfies the Diophantine condition
\begin{align}\label{def:DC}
\mathrm{DC}_{c,a}:=\Big\{\alpha\in \T:\, \|n\alpha\|_{\T}\geq \frac{c}{n (\log n)^a}, \text{ for all } |n|\geq 1,\Big\}
\end{align}
where $c>0$ and $a>1$.
It is well-known that for any $a>1$, $\cup_{c>0}\mathrm{DC}_{c,a}$ is a full measure set.

For an energy $E$, let the Lyapunov exponent $L(E,\varepsilon)$ of the complexified operator $H_{\alpha,\theta+i\varepsilon}$ be defined as usual, see~\eqref{def:Ln_E} below.
As a corner stone of his global cocycle theory~\cite{Global}, Avila exhibited quantization of the slope of the Lyapunov exponent as a convex function of $\varepsilon$: the right derivative $\kappa(E,\varepsilon)$ of $L(E,\varepsilon)/(2\pi)$ in $\varepsilon$ is always an integer, see  Section~\ref{sec:LE}.
Using the quantized acceleration, Avila stratified the spectrum $\sigma(H_{\alpha,\theta})$ into 
\begin{align*}
\sigma(H_{\alpha,\theta})=\bigcup_{\ell \geq 1} \mathcal{S}_{\ell},
\end{align*}
where 
\begin{align*}
\mathcal{S}_{\ell}:=\sigma(H_{\alpha,\theta})\cap \{E:\, \kappa(E,0)=\ell-1\}.
\end{align*}
The set $\mathcal{S}_{\ell}$ is called the $\ell$-th stratum of the spectrum.

For $E\in \mathcal{S}_1$, it is necessary that $L(E,0)=0$ since otherwise the corresponding Schr\"odinger cocycle is uniformly hyperbolic which leads to a contradiction with $E\in \sigma(H_{\alpha,\theta})$.
In fact, $\mathcal{S}_1$ is refereed to as the subcritical regime.
This stratum is usually studied by KAM and reducibility methods, see e.g. \cites{DS,Eliasson,Puig, A_ac, AJ, A_ar_ac}.

For each $\ell\geq 2$, we further divide
\begin{align*}
\mathcal{S}_{\ell}=\mathcal{S}_{\ell}^+\cup \mathcal{S}_{\ell}^0,
\end{align*}
where $\mathcal{S}_{\ell}^+:=\mathcal{S}_{\ell}\cap \{E:\, L(E,0)>0\}$ and $\mathcal{S}_{\ell}^0:=\mathcal{S}_{\ell}\cap \{E:\, L(E,0)=0\}$.
The energies in $\cup_{\ell\geq 2} \mathcal{S}_{\ell}^0$ are referred to as the critical energies, and those in $\cup_{\ell\geq 2} \mathcal{S}_{\ell}^+$ are supercritical.

One of the central examples of quasi-periodic Schr\"odinger operator is the almost Mathieu operator:
\begin{align}\label{def:AMO}
(H^{\mathrm{AMO}}_{\lambda,\alpha,\theta}\phi)_n=\phi_{n+1}+\phi_{n-1}+2\lambda \cos2\pi(\theta+n\alpha)\phi_n.
\end{align}
The supercritical, critical and subcritical regime corresponds to the regions where the coupling constant satisfies $|\lambda|>1$, $|\lambda|=1$, and $0<|\lambda|<1$, respectively.
A special feature for the almost Mathieu operator is that for $|\lambda|>1$, $\sigma(H^{\mathrm{AMO}}_{\lambda,\alpha,\theta})=\mathcal{S}_2^+$, hence the entire spectrum is contained in the first supercritical stratum.

Our analysis focuses on general analytic potentials in the supercritical regime. 
We begin with an immediate corollary of \cite{Global} and~\cite{DGSV} on the Lebesgue measure of strata. 

\begin{theorem}\label{thm:posmeas}
Let $\alpha\in \mathrm{DC}_{c,a}$ for some $c>0$ and $a>1$.
For any $k\ge2$, if 
\begin{equation}\label{eq:unionS}
\bigcup_{j=2}^k \mathcal{S}_j^+\ne\emptyset,\text{\ \  then }\Big| \bigcup_{j=2}^k \mathcal{S}_j^+\Big|>0
\end{equation}
In particular, if $  \mathcal{S}_2^+\ne\emptyset$, then $|  \mathcal{S}_2^+|>0$.
\end{theorem} 

Indeed, let $E_0$ be in the set in~\eqref{eq:unionS}. Then by the upper semicontinuity and quantization of the acceleration~\cite[Theorem 5]{Global}, as well as  the continuity of the Lyapunov exponent \[I_0\subset \Big\{E\in\R \::\: \kappa(E,0) < k+\frac12, \; L(E,0)>0\Big\}\]
for some interval $I_0=(E_0-\delta, E_0+\delta)$. By~\cite[Theorem H]{DGSV}, we see that $| \sigma(H_{\alpha,\theta})\cap I_0|>\frac12\delta$ if $\delta$ is small enough. By~\cite[Theorem 6]{Global} one has  $\kappa(E,0)\geq 1$ on $\sigma(H_{\alpha,\theta})\cap I_0$, hence $\sigma(H_{\alpha,\theta})\cap I_0= \cup_{j=2}^k \mathcal{S}_j^+$, whence the theorem. 

\medskip

We now turn to the main theorem of our work, which characterizes Avila's acceleration via the number of zeros. Throughout we denote annuli by 
$A_r:=\{z\in \C:\, 1/r<|z|<r\}$ and $$N_n(E, \varepsilon):=\#\{z\in \overline{A_{e^{2\pi \varepsilon}}}:\, D_n(z,E)=0\}$$ is the number of zeros relative to $z$ of a determinant in finite volume. 

\begin{theorem}\label{thm:acc=zeros}
Let $L(E)\geq \tau>0$ and let $\varepsilon\in (0,\eta)$ be such that 
$$L(E,\varepsilon)=L(E,0)+2\pi \kappa(E,0)\varepsilon.$$
Then  
for $n>N(\tau,E,f,\eta,\alpha,\gamma)$,
\begin{align}
\left| \frac{1}{2n} N_n (E, \varepsilon/2)-\kappa(E,0)\right|\leq \varepsilon^{-1} n^{-\gamma}.
\end{align}
\end{theorem}
This theorem gives a characterization of the acceleration $\kappa(E,0)$ in terms of the number of zeros of $D_n(z,E)$ in a neighborhood of the unit circle.
Theorem \ref{thm:acc=zeros} is part of the more general Theorem \ref{thm:zero_count} (see also Theorem \ref{thm:Riesz_un}), where the characterization of $\kappa(E,\varepsilon)-\kappa(E,\varepsilon^-)$, for non-zero $\varepsilon$, is given in terms of zeros of $D_n(z,E)$ near the circle with radius $e^{2\pi \varepsilon}$.
An analogous result for the Riesz mass of $\log \|M_n(z,E)\|$, where $M_n$ is the transfer matrix~\eqref{def:Mn}, is given in Theorem \ref{thm:Riesz_vn}.

A key new ingredient in the proof is a Riesz representation of subharmonic function on the annulus, compared to the one on disks that were used to cover the annulus in \cite{GS1,GS2}.
It turns out the annulus version of the representation is more naturally connected to the Lyapunov exponent, hence to the acceleration. 
As one of the applications of Theorem \ref{thm:acc=zeros}, we improve on the known regularity of the integrated density of states (IDS) as a function of the energy. 

\begin{theorem}\label{thm:IDS}
Assume the potential function $f$ is analytic and fix $\alpha\in \mathrm{DC}_{c,a}$ for some $c>0$ and $a>1$.  Then for any $\ell\ge2$, on the open set $\bigcup_{j=1}^\ell \mathcal{S}_j^+$ the IDS of $H_{\alpha,\theta}$ is $\beta$-H\"older  continuous with any $0<\beta<\frac{1}{2(\ell-1)}$.
\end{theorem}

This result proves the conjecture by You in \cite{You} that the H\"older exponent for IDS is at least $1/(2\kappa(E,0))$ (but we cannot decide equality of the exponent to this ratio). Note that $\bigcup_{j=1}^\ell \mathcal{S}_j^+$ includes gaps in the spectrum on which the IDS is constant. But by Theorem~\ref{thm:posmeas}, if this set of energies is nonempty, then it contains a positive measure subset of the spectrum. 
The proof of Theorem \ref{thm:IDS} is not self-contained. In fact, Theorem~\ref{thm:IDS} follows by combining the proof of \cite[Theorem 1.1]{GS2}  with the characterization of the acceleration (and thus the stratum) in Theorem~\ref{thm:acc=zeros}. Also, recall from~\cite{GS2} that the IDS is Lipschitz off a zero-measure set of energies. 

\begin{remark}
In \cite[Theorem 1.1]{GS2}, it was proved that for $f(\theta) = \sum_{n=-k_0}^{k_0}  c_n e^{2\pi i n\theta}$, a trigonometric polynomial of degree $k_0\ge1$ and assuming $\omega\in {DC}_{c,a}$ and 
$L(E,0)>0$  the IDS is $\beta$-H\"older continuous for any $\beta<\frac{1}{2k_0}$. 
Since for trigonometric $f$ one has $\kappa(E,0)\leq k_0$, Theorem \ref{thm:IDS} is a refinement of \cite[Theorem 1.1]{GS2}.
\end{remark}

\begin{remark}
For the almost Mathieu operator and $\alpha\in \mathrm{DC}$ (see \eqref{def:DC_1}), the IDS was proved to be $1/2$-H\"older if $|\lambda|\neq 0,1$ \cite{AJ}.
The $1/2$-H\"older exponent was also proved for Diophantine $\alpha$ and $f=\lambda\, g$ with $g$ being analytic and the coupling constant $|\lambda|$ being small, first in the perturbative regime (smallness depends on $\alpha, g$) in \cite{Amor}, and then in the non-perturbative regime \cite{AJ} (with dependence on $\alpha$ removed).
For quasi-periodic long-range operators with large trigonometric polynomial potentials and Diophantine frequencies, the H\"older exponent in \cite{GS2} was improved recently in the perturbative regime in \cite{GYZ}.
These results are proved using the reducibility method.
\end{remark}

Our next application concerns Anderson localization (pure point spectrum with exponentially decaying eigenfunctions).
Anderson localization for quasi-periodic operators was first studied perturbatively \cite{FSW,Sinai,E97} where $f=\lambda\, g$ with large $\lambda$ whose largeness depends on both $\alpha$ and $g$.
The first non-perturbative Anderson localization was obtained by Jitomirskaya for the almost Mathieu operator \cite{J94,J99} for any $|\lambda|>1$ and all Diophantine frequencies.
In \cite{BG}, Bourgain and Goldstein proved Anderson localization for the general analytic potential in the supercritical regime for an implicit full measure set of $\alpha$.
In this paper we show that quasi-periodic operators with even potentials in the first supercritical stratum $\mathcal{S}_2^+$ is always Anderson localized for all Diophantine frequencies. This result thus includes that of the supercricitical almost Mathieu operator as a special case.
Let 
\begin{align}\label{def:Theta}
\Theta_{c',b}:=\bigcap_{k\geq 1}\bigcup_{|n|\geq k} \Big\{\theta\in \T:\, \|2\theta+n\alpha\|_{\T}<\frac{c'}{|n|^b}\Big\},
\end{align}
where $c'>0$ and $b>1$.
Clearly the complement $(\Theta_{c',b})^c$ is a full measure set.
\begin{theorem}\label{thm:loc}
Assume the potential function $f$ is even. For $\alpha\in \mathrm{DC}_{c,a}$ for some $c>0$ and $a>1$, and $\theta\in (\Theta_{c',b})^c$ for some $c'>0$ and $b>1$, $H_{\alpha,\theta}$ has pure point spectrum with exponentially decaying eigenfunctions in $\mathcal{S}_2^+$.
\end{theorem}
Recall that by Theorem \ref{thm:posmeas}  if $\mathcal{S}_2^+$ is non-empty, then this set has to be of positive measure.
As a corollary of Theorem \ref{thm:posmeas}, we can prove Anderson localization for 
analytic perturbations of the supercritical almost Mathieu operator, which has recently been studied by Ge-Jitomirskaya-Zhao in \cite{GJZ}.

\begin{corollary}\label{cor:loc}
Let the perturbed almost Mathieu operator be defined as
\begin{align}\label{def:pert_AMO}
(H_{\lambda,\alpha,\theta,g,\varepsilon} \phi)_n=\phi_{n+1}+\phi_{n-1}+(2\lambda\cos 2\pi(\theta+n\alpha)+\varepsilon g(\theta+n\alpha)) \phi_n,
\end{align}
where $g \in C^{\omega}_{\T_{\eta}}$ is even and real-valued.
For $|\lambda|>1$, there exists $\varepsilon_1=\varepsilon_1(\lambda, \|g\|_{\T_{\eta}})>0$ such that if $|\varepsilon|<\varepsilon_1$, and $\alpha\in \mathrm{DC}_{c,a}$ for some $c>0$ and $a>1$, and $\theta\in (\Theta_{c',b})^c$ for some $c'>0$ and $b>1$,
$H_{\lambda,\alpha,\theta,g,\varepsilon}$ has pure point spectrum with exponentially decaying eigenfunctions.
\end{corollary}
This corollary follows from combining Theorem \ref{thm:loc} with the following lemma of Avila.
\begin{lemma}[Lemma 25, \cite{Global}]
For $|\lambda|>1$ and $\alpha\in \R\setminus \Q$, there exists $\varepsilon_1= \varepsilon_1(\lambda,\|g\|_{\T_{\eta}})>0$ such that for $|\varepsilon|<\varepsilon_1$, $\sigma(H_{\lambda,\alpha,\theta,g,\varepsilon})=\mathcal{S}_2^+$.
\end{lemma}

\begin{remark}\label{rm:DC}
It is possible to prove Theorem \ref{thm:loc} and Corollary \ref{cor:loc} for $\alpha$ satisfying a weaker Diophantine condition:
\begin{align}\label{def:DC_1}
\mathrm{DC}:=\bigcup_{c>0}\bigcup_{\rho>1}\left\{\alpha\in \T:\, \|n\alpha\|_{\T}\geq \frac{c}{|n|^{\rho}}, \text{ for any } n\neq 0\right\}.
\end{align}
Replacing our $\mathrm{DC}_{c,a}$ with the weaker condition will lead to less sharp large deviation estimates in Section \ref{sec:LDT}, but these will still suffice for the proof of Anderson localization.
\end{remark}

Next let us comment on the proof of Theorem \ref{thm:loc}.
The key to prove Anderson localization is to eliminate ``double resonance'', which roughly speaking occurs when both $D_n(\theta,E)$ and $D_n(\theta+k\alpha,E)$ are close to~$0$ for $k\simeq n^C$.
Jitomirskaya's \cite{J94,J99} argument for the almost Mathieu operator uses crucially that the potential is cosine.
She proved $D_n(\theta-\frac{n-1}{2}\alpha, E)$ is a polynomial of $\cos(2\pi \theta)$ of degree at most $n$, and used the Lagrange interpolation formula to eliminate double resonances.
Bourgain and Goldstein \cite{BG} developed a large deviation estimate for the set on which the norm of the transfer matrix is small, and then used semi-algebraic tools to study the complexity of this set.
In order to eliminate double resonances, they then devised a ``steep line'' argument, which requires deleting an implicit zero measure of~$\alpha$'s. 
In this paper, on the one hand we appeal to the large deviation estimates in \cite{GS2} for the sub-exponentially small measure of the deviation set $\mathcal{B}_n$ (see \eqref{def:Bn}), on the other hand we combine the $\kappa=1$ case of Theorem~\ref{thm:acc=zeros} with a Cartan set argument to show that $\mathcal{B}_n$ consists of at most $2n+o(n)$ many intervals.
Using that the potential is even, the $2n+o(n)$ intervals form $n+o(n)$ pairs under reflection $\theta\to -\theta$.
The arithmetic conditions on $\alpha,\theta$ prevent that two points  in the trajectory $\{\theta+j\alpha\}$ fall into the same pair (at the appropriate scales).
This eliminates double resonances. 
It is worth mentioning that although the $\kappa=1$ case of Theorem~\ref{thm:acc=zeros} ensures that $D_n(z,E)$ has at most $2n+o(n)$ many zeros, $D_n(\theta-\frac{n-1}{2}\alpha,E)$ is not a polynomial of degree at most $2n+o(n)$ in $\cos(2\pi\theta)$, so the Lagrange interpolation argument as in \cites{J94,J99} does not apply.

The evenness assumption on the potential is crucial for our argument. It is not clear if one can improve on~\cite{BG} without making this symmetry assumption in the sense that there is an explicit condition on the phases~$\theta$ which guarantees localization. 
Recently, Forman and VandenBoom \cite{FV} proved Anderson localization perturbatively for large cosine-like potentials without evenness assumption, 
thus removing the symmetry requirement from the classical work by Fr\"ohlich, Spencer, and Wittwer~\cite{FSW}. However, while~\cite{FSW} eliminate resonant phases via an arithmetic condition, \cite{FV} do so by an implicit procedure.

Finally, let us mention that it is an interesting question as to whether the results in Theorems~\ref{thm:IDS} and \ref{thm:loc} hold for Liouville frequencies. 
Related questions have been addressed for the almost Mathieu operator in e.g. \cite{AD,A_ac,AYZ,JL}.
For general analytic potentials, large deviation estimates have been developed for some Liouville $\alpha$'s in \cite{YZ,HZ}, leading to the proof of H\"older continuity of the IDS, albeit with non-sharp exponent.
It would be interesting to see if one can combine the techniques in \cite{YZ,HZ} with ours.

This paper is organized as follows. In Section~\ref{sec:pre} we collect some preliminary results; in Section~\ref{sec:Riesz} we develop the Riesz representation of subharmonic functions on an annulus; in Sections~\ref{sec:Riesz_un} and \ref{sec:Riesz_vn} we apply the Riesz representation to the functions $u_n$ and $v_n$ and establish the zero count of  Theorem~\ref{thm:acc=zeros}. Finally,  Anderson localization as in Theorem~\ref{thm:loc} is obtained in Section~\ref{sec:AL}, and Theorem~\ref{thm:IDS} on H\"older continuity of the IDS is proved in Section~\ref{sec:IDS}.

\section{Preliminaries}\label{sec:pre}
Notations: Let $A_R:=\{z\in \C:\, 1/R<|z|<R\}$ be the annulus and $\mathcal{C}_R:=\{z\in \C:\, |z|=R\}$ be the circle.
For a set $U\subset \R$, let $|U|$ be its Lebesgue measure. 

\subsection{Determinant and transfer matrix}

Since $f\in C_{\eta}^{\omega}(\T)$,  the determinant $D_n(z,E)$ is holomorphic in the annulus $A_{e^{2\pi \eta}}$.
Moreover, as $f$ is real-valued for $\theta\in \R$, one has $D_n(z,E)=\overline{D_n(1/\overline{z}, E)}$ for $z\in \mathcal{C}_1$, hence
\begin{align}\label{eq:Pn_real}
D_n(z,E)=\overline{D_n(1/\overline{z}, E)}, \text{ for } z\in A_{e^{2\pi \eta}}.
\end{align}
In particular, we note the following fact that will be used many time throughout the paper:
\begin{fact}\label{fact:zero}
If $w\notin \mathcal{C}_1$ is a zero of $D_n(z,E)$, then $1/\overline{w}$ is also a zero.
\end{fact}

Let 
\begin{align}\label{def:un}
u_n(z,E):=\frac{1}{n}\log |D_n(z,E)|.
\end{align}
and as usual,  
\begin{align}\label{def:Mn}
M(\theta, E)=
\left(\begin{matrix}
E-f(\theta)\ &-1\\
1 &0
\end{matrix}\right)
\end{align}
denotes the transfer matrix, and 
\begin{align*}
M_n(\theta,E):=M(\theta+(n-1)\alpha, E)\cdots M(\theta,E)
\end{align*}
be the n-step transfer matrix.
Let 
\begin{align}\label{def:vn}
v_n(z,E):=\frac{1}{n}\log \|M_n(z,E)\|.
\end{align}
Both $u_n$ and $v_n$ are subharmonic functions in the annulus $A_R$. 
The starting point of our study is an effective Riesz representation of subharmonic functions in $A_R$, which we develop in Section \ref{sec:Riesz}.

Recall the well-known connection between $D_n$ and $M_n$, viz. 
\begin{align}\label{eq:Dn_Mn}
M_n(z,E)=
\left(\begin{matrix}
D_n(z,E)\, &-D_{n-1}(ze^{2\pi i \alpha},E)\\
D_{n-1}(z,E) &-D_{n-2}(ze^{2\pi i\alpha},E)
\end{matrix}\right).
\end{align}
This implies that 
\begin{align}\label{eq:un<vn}
u_n(z,E)\leq v_n(z,E),
\end{align}
which will be used multiple times in this paper.

\subsection{Lyapunov exponent and acceleration}\label{sec:LE}
Let 
\begin{align}\label{def:Ln_E}
L_n(E,\varepsilon):=\int_0^{1} v_n(e^{2\pi i(\theta+i\varepsilon)}, E)\, \mathrm{d}\theta,
\end{align}
and
\begin{align}\label{def:LE}
L(E,\varepsilon):=\lim_{n\to\infty} L_n(E,\varepsilon).
\end{align}
The Lyapunov exponent is even in $\varepsilon$ due to $f$ being real-valued, and is a convex function in $\varepsilon$ by Hadamard's three circle theorem. We define its derivative from the right as 
\begin{align}\label{def:acc}
\kappa(E,\varepsilon):=\lim_{\delta\to 0^+} \frac{L(E, \varepsilon+\delta)-L(E,\varepsilon)}{2\pi \delta}.
\end{align}
The function $\kappa(E,\varepsilon)$ is referred to as the {\em acceleration} of the Lyapunov exponent.

In \cite{Global}, Avila proved the following:
\begin{theorem}\label{thm:acc_N}\cite[Theorem 5]{Global}
For any irrational $\alpha$ and any $E,\varepsilon$, the acceleration $\kappa(E,\varepsilon)$ is always an integer.
\end{theorem}

Using the quantized acceleration, Avila further introduced the stratification of the spectrum: $X_1:=\sigma(H_{\alpha,\theta})$, and $X_j$ with $j\geq 2$ is defined as
\begin{align*}
X_j:=\{E\in X_1:\, \kappa(E,0)\geq j-1\}.
\end{align*}
The set 
\begin{align}\label{def:Sj}
\mathcal{S}_j:=X_j\setminus X_{j+1}=\{E\in \sigma(H_{\alpha,\theta}):\, \kappa(E,0)=j-1\}
\end{align}
is called the $j$-th stratum of the stratification.
It is clear that $\{\mathcal{S}_j\}_{j=1}^{\infty}$ are pairwise disjoint and
\begin{align*}
\bigcup_{j\geq 1} \mathcal{S}_j=X_1=\sigma(H_{\alpha,\theta}).
\end{align*}

\subsection{Large deviation estimates}\label{sec:LDT}
Our analysis depends crucially on the large deviation estimates (and their consequences) for $u_n$ and $v_n$ developed in a series of papers \cite{BG,GS1,GS2}, in combination with the avalanche principle. 
We list the results that are crucial to our analysis below.
Recall that the potential $f\in C^{\omega}(\T_{\eta})$.
The first one is a uniform upper bound for $v_n$, the logarithm of the norm of the transfer matrix.
\begin{lemma}\label{lem:vn_upper}\cite[Proposition 4.3]{GS2}
Assume $L(E,0)>\tau>0$. 
Then for all $n\geq 1$, and $0\leq |\varepsilon|<\eta/2$,
\begin{align*}
\sup_{\theta\in \T} v_n(e^{2\pi i(\theta+i\varepsilon)},E)\leq L_n(E,\varepsilon)+C\frac{(\log n)^{C_0}}{n}.
\end{align*}
for some constants $C=C(\|f\|_{\eta},\eta,c,a,\tau,E)$ and $C_0=C_0(a)$, with $c$, $a$ as in \eqref{def:DC}.
\end{lemma}

The next is a quantitative convergence rate of $L_n(E,\varepsilon)$ to $L(E,\varepsilon)$.
\begin{theorem}\label{thm:Ln_L}\cite[Theorem 5.1]{GS1}
If $L(E,0)>\tau>0$, then for $n\geq 1$, one has that for $|\varepsilon|<\eta/2$,
\begin{align*}
0\leq L_n(E,\varepsilon)-L(E,\varepsilon)\leq \frac{C}{n},
\end{align*}
where $C=C(\tau, E, f, \eta, \alpha)$.
\end{theorem}

The next is an average lower bound for $u_n$.
\begin{lemma}\label{lem:un_lower}\cite[Lemma 2.10]{GS2}
If $L(E)>\tau>0$, there exists a positive constant $\gamma_1>0$ such that for $|\varepsilon|<\eta/2$, and all $n>N(\tau,E,f,\eta,\alpha)$,
\begin{align*}
\int_0^{1} u_n(e^{2\pi i(\theta+i\varepsilon)},E)\, \mathrm{d}\theta>L_n(E,\varepsilon)-n^{-\gamma_1}.
\end{align*}
\end{lemma}

The final tool is a large deviation estimate for $u_n$.
\begin{lemma}\label{lem:un_LDT}\cite[Proposition 2.11]{GS2}
If $L(E)>\tau>0$, there exists a positive constant $\gamma_2>0$ such that for $n>N(\tau,E,f,\eta,\alpha)$,
\begin{align*}
\left|\{\theta\in \T:\, u_n(e^{2\pi i\theta},E)<L_n(E)-n^{-\gamma_2}\}\right|<e^{-n^{\gamma_2}}.
\end{align*}
\end{lemma}

\subsection{Eigenfunction expansion}
Let $\phi$ be a solution to $H_{\alpha,\theta}\phi=E\phi$, where $H_{\alpha,\theta}$ is the Schr\"odinger operator in \eqref{def:H_operator}.
Then the following eigenfunction expansion holds: let $y\in [\ell_1, \ell_2]$,
\begin{align}\label{eq:Green_exp}
\phi_y=-\frac{D_{\ell_2-y}(\theta+(y+1)\alpha,E)}{D_{\ell_2-\ell_1+1}(\theta+\ell_1\alpha,E)}\phi_{\ell_1-1}-\frac{D_{y-\ell_1}(\theta+\ell_1\alpha,E)}{D_{\ell_2-\ell_1+1}(\theta+\ell_1\alpha,E)}\phi_{\ell_2+1}.
\end{align}
This identity is a consequence of Cramer's rule.

\section{Effective Riesz representation for the annulus}\label{sec:Riesz}

This section develops some basic potential theory on annuli. We begin with the Green's  function on an annulus, which is standard. 
 
\subsection{Green's function for the annulus.}
It is worth pointing out that we never use the explicit expression of the Green's function, rather only the fact that it can be decomposed as in \eqref{eq:GRH*} below. For the sake of completeness, we state the precise Green's kernel, which can be derived  by the method of images. 

\begin{lemma}\label{lem:Green_AR}
The Green's function on the annulus $A_R$ is given by
\begin{align}\label{def:G}
G_R(z,w)=\frac{\log( |z|/ R) \log (|w|/ R)}{4\pi\log R} +K_R(z,w),
\end{align}
where
\begin{align}\label{eq:tildeG_prod}
K_R(z,w)=\frac{1}{2\pi}\log \left(|z/R-w/R| \cdot \frac{\prod_{k=1}^{\infty} |1-\frac{1}{R^{4k}}\frac{z}{w}| \cdot |1-\frac{1}{R^{4k}} \frac{w}{z}|}{\prod_{k=1}^{\infty} |1-\frac{1}{R^{4k-2}}w\overline{z}|\cdot |1-\frac{1}{R^{4k-2}} \frac{1}{\overline{z}w}|}\right).
\end{align}
The Green's function is symmetric and invariant under rotations: $G_R(z,w)=G_R(w,z)$ and $G_R(z,w)=G_R(e^{i\phi} z,e^{i\phi} w)$. 
\end{lemma}
\begin{proof}
The symmetry properties are evident from the formula. Second, by inspection 
\begin{align} 
G_R(z,w) &= \frac{1}{2\pi}\log |z-w| + H_R(z,w), \quad z\in A_R, w\in \overline{A_R} \label{eq:GRH*} \\
\Delta_z H_R (z,w) &= 0 \nonumber
\end{align}
whence $\Delta_z  G(z,w) =\delta_0(z-w)$. 
Third, if $w=R$, then $$G_R(z,w)=   K_R(z,w) =     \frac{1}{2\pi}\log \left(\frac{|z/R-1| }{  |1-\frac{1}{R}\overline{z}|}\right) =0 $$
This holds for all $|w|=R$ by rotational invariance of $G_R$. If $|w|=R^{-1}$, then
\[
G_R(z,1/R) =- \frac{1}{2\pi}\log (|z|/R) +  \frac{1}{2\pi}\log\left( \frac{|z/R-1/R^2|}{|1-1/(R\bar{z})|}\right) =0
\]
and this again holds for all $|w|=R^{-1}$ by rotational invariance. The lemma follows from the uniqueness of the Green's function. 
\end{proof}

We evaluate the integral of the Green's function along the circle $z\in \mathcal{C}_r$ with $1/R\leq r\leq R$.
These integrals will be used later in Section \ref{sec:Riesz_un} to estimate the number of zeros of $D_n(z,E)$, and the Riesz mass of $v_n(z,E)$ in Section \ref{sec:Riesz_vn}.

\begin{lemma}\label{lem:int_Green}
For $1/R\leq r\leq R$ and $w\in A_R$, we have
\begin{align}\label{eq:int_Green}
I(\log r, \log R, w) &:=2\pi \int_0^{1}  G_R (re^{ 2\pi i\theta}, w)\, \mathrm{d}\theta \\
& =(2\log R)^{-1}
\begin{cases}
 \log (rR) \log|w/R|, \text{ if } |w|\geq r\\
\\
\log( r/ R) \log |w R|, \text{ if } |w|<r.
\end{cases}
\end{align}
\end{lemma}
\begin{proof}
\eqref{eq:GRH*} yields
\begin{align}
2\pi \int_0^{1}  G_R (re^{ 2\pi i\theta}, w)\, \mathrm{d}\theta &= \int_0^{2\pi} \frac{1}{2\pi}\log |re^{ i\theta} -w|\, d\theta + \int_0^{2\pi}   H_R(re^{ i\theta},w) \, \mathrm{d}\theta  \\
&=: J_1(w)+J_2(w)   \nonumber
\end{align}
where $J_1(w)=\log|w|$ if $|w|\ge r$ and $J_1(w)=\log r$ if $|w|\le r$.  $J_2(w)$ is harmonic in $w\in A_R$, radial,  and continuous on~$\overline{A_R}$. Thus, $J_2(w)=a\log|w|+b$. Setting $|w|=R$, respectively $|w|=R^{-1}$ shows that 
\[
J_2(w) = \frac{\log(r/R)}{2\log  R} \log(|w|/R) - \log R
\]
which implies \eqref{eq:int_Green}.  
\end{proof}

\subsection{Effective Riesz representation}

We now turn to the basic Riesz representation of subharmonic functions. In contrast to~\cite{GS1,GS2}, which analyzed supercricial cocycles by means of potential theory on small disks, here we conduct this analysis globally on annuli. 

\begin{lemma}\label{lem:Riesz}
Let $v$ be a subharmonic function in a neighborhood of $\overline{A_R}$, and assume $v|_{\partial A_R}$ is a continuous function.
Let $G_R$ be the Green's function for $A_R$, as  in \eqref{def:G}.
There exists a positive finite measure $\mu$ on $A_R$, and a harmonic function $h_R$ on $A_R$, such that
\begin{align*}
v(w)=\int_{A_R} 2\pi G_R(z,w)\, \mu(\mathrm{d}z)+h_R(w),
\end{align*}
where 
\begin{align}\label{eq:Poisson}
h_R(w)=\int_{\partial A_R} v(z)\,\nu(w,A_R)( \mathrm{d}z),
\end{align}
where $\nu(w,A_R)$ is the harmonic measure of $A_R$ with pole at $w$.
In particular, 
\begin{align}\label{eq:h=v}
h_R(z)=v(z), \text{ for } z\in \partial A_R.
\end{align}
\end{lemma}
\begin{proof}
Without loss of generality, we may assume that $v$ is smooth. If this is not the case, we convolve $v$ with
a radial nonnegative mollifier. The submean property then guarantees monotone convergence. We skip these technical details.

By Green's second identify, with the Green's function $G_R$ defined in \eqref{def:G}, 
\begin{align*}
v(w)-\int_{A_R} G_R(z,w) \Delta v(z)\, m(\mathrm{d}z)=\int_{\partial A_R} v(z) \frac{\partial G_R}{\partial n_z}(z,w)\, \sigma(\mathrm{d}z),
\end{align*}
where $m$ is Lebesgue measure and $\sigma$ is the (unnormalized) arclength measure on $\partial A_R$. 
Since $v$ is smooth and subharmonic, $\Delta v$ is a non-negative, continuous function, and defines a positive measure $2\pi \mu=\Delta v\, m$. Therefore
\begin{align}
v(w)=\int_{A_R} 2\pi G_R(z,w)\, \mu(\mathrm{d}z)+h_R(w),
\end{align}
where 
\begin{align}
h_R(w):=&\int_{\partial A_R} v(z) \frac{\partial G_R}{\partial n_z}(z,w)\, \sigma(\mathrm{d}z)\\
=&\int_{\partial A_R} v(z)\,\nu(w,A_R)(\mathrm{d}z)
\end{align}
is the harmonic part.
\end{proof}

\begin{remark}
By the maximum principle, we have
\begin{align}\label{eq:max_h}
\sup_{w\in A_R} h_R(w)\leq \max_{z\in \partial A_R} v(z).
\end{align}
The  harmonic measure of the annulus $A_R$ can be computed explicitly from $G_R$ above, but we have no need for that. 
We only require a basic well-known bound on the density of the harmonic measure, namely that it is controlled by the inverse distance to $\partial A_R$. 
\end{remark}

\section{Riesz representation for $u_n$}\label{sec:Riesz_un}
Recall that $D_n(z,E)$ is a holomorphic function in $A_{e^{2\pi \eta}}$. 
For $0\leq \varepsilon<\eta$, let $$N_n(E,\varepsilon):=\#\{z\in \overline{A_{e^{2\pi \varepsilon}}}:\, D_n(z,E)=0\},$$ in particular $$N_n(E,0)=\#\{z\in \mathcal{C}_1:\, D_n(z,E)=0\}.$$

\subsection{Avila's acceleration via a zero count}
We establish the following more general form of Theorem \ref{thm:acc=zeros}.
\begin{theorem}\label{thm:zero_count}
For some energy $E\in \R$ assume that $L(E,0)\geq \tau>0$. 
Assume further that $0< \varepsilon_0<\eta$ is such that there exists $\varepsilon_1>0$, $[\varepsilon_0-\varepsilon_1,\varepsilon_0+\varepsilon_1]\subset [0,\eta)$, $D_n(z,E)$ is zero-free on $\partial A_{e^{2\pi i (\varepsilon_0+\varepsilon_1)}}$, and
\begin{align}\label{eq:zero_count_asp}
\kappa(E,\varepsilon_0+\varepsilon)\equiv \kappa,
\end{align}
for $|\varepsilon|<\varepsilon_1$.
Then for the constant $\gamma_1>0$ in Lemma \ref{lem:un_lower}, for $n>N(\tau,E,f,\eta,\alpha,\gamma_1)$ and some absolute constant $C_1>0$, 
\begin{align*}
\left|\frac{1}{2n} N_n(E,\varepsilon_0+\frac{1}{3}\varepsilon_1)-\kappa\right|\leq C_1\varepsilon_1^{-2} n^{-\gamma_1}.
\end{align*} 
\end{theorem}
\begin{remark}
Suppose that for  $\varepsilon_0\in (0,\eta)$  there exists $\varepsilon_1>0$ such that $(\varepsilon_0-\varepsilon_1, \varepsilon_0+\varepsilon_1)\subset [0,\eta)$ and for $0<\varepsilon\leq \varepsilon_1$,
\begin{align*}
\kappa(E,\varepsilon_0-\varepsilon)\equiv \kappa_1< \kappa_2\equiv \kappa(E,\varepsilon_0+\varepsilon).
\end{align*}
Then Theorem \ref{thm:zero_count} yields for $n>N(\tau,E,f,\eta,\alpha,\gamma_1)$,
\begin{align*}
\left|\frac{1}{2n} \left(N_n(E,\varepsilon_0+\frac{1}{3}\varepsilon_1)-N_n(E,\varepsilon_0-\frac{1}{3}\varepsilon_1)\right)-(\kappa_2-\kappa_1)\right|\leq C_1\varepsilon_1^{-2} n^{-\gamma_1}.
\end{align*}
Hence one can also characterize the change of slopes of the piece-wise linear function $L(E,\varepsilon)$ in terms of zero counts of $D_n(z,E)$.
\end{remark}
\begin{proof}
In the proof we shall omit the dependence of various parameters on $E$ for simplicity.
By the assumption \eqref{eq:zero_count_asp} and Theorem \ref{thm:Ln_L}, one has for $0\leq |\varepsilon|\leq \varepsilon_1$,
\begin{align}\label{eq:zero_count_LE}
L(\varepsilon_0)+2\pi \kappa \varepsilon\leq  L(\varepsilon_0+\varepsilon)\leq  L_n(\varepsilon_0+\varepsilon)\leq L(\varepsilon_0+\varepsilon)+\frac{C}{n}=L(\varepsilon_0)+2\pi\kappa\varepsilon+\frac{C}{n}.
\end{align}
Let $\tilde{R}:=e^{2\pi(\varepsilon_0+\varepsilon_1)}$ and $\tilde{N}:=N_n(\varepsilon_0+\varepsilon_1)$.
Denote by $w_1,w_2,...,w_{\tilde{N}}$ the zeros of $D_n(z)$ in $A_{\tilde{R}}$.
Define
$$G_{\tilde{R},n}(z):=\frac{1}{n}\sum_{k=1}^{\tilde{N}} G_{\tilde{R}}(z,w_k),$$
where $G_R(\cdot,\cdot)$ is the Green's function in \eqref{def:G}.
Lemma~\ref{lem:Riesz}, with $R=\tilde{R}$, applied to $u_n$ yields 
\begin{align}\label{eq:thm_zeros}
u_n(z)=2\pi G_{\tilde{R},n}(z)+h_{\tilde{R},n}(z),
\end{align}
First, we estimate the harmonic part.
\begin{lemma}\label{lem:h_un}
With the constant $\gamma_1$ in Lemma \ref{lem:un_lower} and the constant $C_0$ in Lemma \ref{lem:vn_upper},  for $z\in A_{e^{2\pi \varepsilon}}$, $0\leq \varepsilon< \varepsilon_0+\varepsilon_1$, and $n>N(\tau,E,f,\eta,\alpha)$ 
\begin{align}\label{eq:h_constant_un}
L_n(E,\varepsilon_0+\varepsilon_1)-\frac{C}{\tilde{R}-e^{2\pi \varepsilon}} n^{-\gamma_1}\leq  h_{\tilde{R},n}(z,E)\leq L_n(E,\varepsilon_0+\varepsilon_1)+C\frac{(\log n)^{C_0}}{n},
\end{align}
\end{lemma}
\begin{proof}
In the proof we shall omit the dependence on $E$ for simplicity.
Note that the harmonic part satisfies $h_{\tilde{R},n}=u_n$ on $\partial A_{\tilde{R}}$, due to \eqref{eq:un_sym}.
By Lemma \ref{lem:vn_upper} and \eqref{eq:un<vn}, one has that for $r=\tilde{R}$ or $1/\tilde{R}$ and $n$ large, uniformly in $\theta$,
\begin{align*}
h_{\tilde{R},n}(re^{2\pi i\theta})=u_n(re^{2\pi i\theta})\leq v_n(re^{2\pi i\theta})\leq L_n(\varepsilon_0+\varepsilon_1)+C\frac{(\log n)^{C_0}}{n}.
\end{align*}
Hence by the maximum principle \eqref{eq:max_h},  
\begin{align}\label{eq:upper_h}
h_{\tilde{R},n}(z)\leq L_n(\varepsilon_0+\varepsilon_1)+C\frac{(\log n)^{C_0}}{n}, \text{ for } z\in \overline{A_{\tilde{R}}}.
\end{align}
We also have by Lemma \ref{lem:un_lower} that for $n$ large enough,
\begin{align}\label{eq:int_h_lower}
\int_0^1 h_{\tilde{R},n}(\tilde{R}e^{2\pi i\theta})\, \mathrm{d}\theta=\int_0^{1} u_n(\tilde{R}e^{2\pi i\theta})\, \mathrm{d}\theta\geq L_n(\varepsilon_0+\varepsilon_1)-\frac{1}{n^{\gamma_1}}.
\end{align}
Let 
\begin{align}\label{def:tilde_h}
\tilde{h}_{\tilde{R},n}(z):=L_n(\varepsilon_0+\varepsilon_1)+C\frac{(\log n)^{C_0}}{n}-h_{\tilde{R},n}(z)\geq 0,
\end{align} 
where we invoked \eqref{eq:upper_h}.
In view of \eqref{eq:int_h_lower},  for $n$ large,
\begin{align}\label{eq:int_th_lower}
\int_0^{1} \tilde{h}_{\tilde{R},n}(\tilde{R}e^{2\pi i\theta})\, \mathrm{d}\theta\leq \frac{2}{n^{\gamma_1}}.
\end{align}
By \eqref{eq:Poisson} and \eqref{eq:int_th_lower}, and the well-known estimate on the harmonic measure
$$0\leq \frac{\mathrm{d}\nu(w, A_R)(z)}{\mathrm{d}\sigma(z)}\leq C(\mathrm{dist}(w,\partial A_R))^{-1},$$ 
with arclength measure $\sigma$, 
one has that for $z\in A_{e^{2\pi \varepsilon}}$, with $0\leq \varepsilon<\varepsilon_0+\varepsilon_1$,
\begin{align*}
0\leq \tilde{h}_{\tilde{R},n}(z)\leq \frac{C}{\tilde{R}-e^{2\pi\varepsilon}} \int_{0}^{1} \tilde{h}_{\tilde{R},n}(\tilde{R}e^{2\pi i\theta})\, \mathrm{d}\theta\leq \frac{C}{\tilde{R}-e^{2\pi\varepsilon}} n^{-\gamma_1}.
\end{align*}
This combined with \eqref{eq:upper_h} yields the claimed result.
\end{proof}

Next, we evaluate the integrals of $G_{\tilde{R},n}(z)$ along circles.
For $1\leq r\leq \tilde{R}$,  
\begin{align*}
I_n(\log r, \log \tilde{R}):=&\int_0^{1} 2\pi G_{\tilde{R},n}(r e^{2\pi i\theta})\, d\theta\\
=&\frac{1}{n}\sum_{k=1}^{\tilde{N}} I(\log r, \log \tilde{R}, w_k),
\end{align*}
where $I(\log r,\log \tilde{R}, w)$ is defined as in~\eqref{eq:int_Green}.
By \eqref{eq:int_Green},
\begin{align*}
I_n(\log r, \log \tilde{R})
=&\frac{1}{n}\sum_{|w_k|>r}\frac{\log r+\log \tilde{R}}{2\log \tilde{R}} \log \frac{|w_k|}{\tilde{R}}\\
&+\frac{1}{n}\sum_{|w_k|<1/r}\frac{\log r-\log \tilde{R}}{2\log \tilde{R}} \log |w_k \tilde{R}|
+\frac{1}{n}\sum_{1/r\leq |w_k|\leq r} \frac{\log r-\log \tilde{R}}{2\log \tilde{R}} \log |w_k \tilde{R}|.
\end{align*}
In view of Fact \ref{fact:zero}, 
\begin{align}\label{eq:int_1}
&\sum_{|w_k|>r}\frac{\log r+\log \tilde{R}}{2\log \tilde{R}} \log \frac{|w_k|}{\tilde{R}}+\sum_{|w_k|<1/r}\frac{\log r-\log \tilde{R}}{2\log \tilde{R}} \log |w_k \tilde{R}|\notag\\
=&\sum_{|w_k|>r}\left(\frac{\log r+\log \tilde{R}}{2\log \tilde{R}} \log \frac{|w_k|}{\tilde{R}}+\frac{\log r-\log \tilde{R}}{2\log \tilde{R}} \log \frac{\tilde{R}}{|w_k|}\right)\notag\\
=&\sum_{|w_k|>r} \log \frac{|w_k|}{\tilde{R}}
=-\sum_k \int_r^{\tilde{R}} \chi_{[|w_k|, \tilde{R})}(x)\, \frac{\mathrm{d}x}{x} \notag\\
=&-\int_r^{\tilde{R}} \#\{w_k:\, r<|w_k|\leq x\}\frac{\mathrm{d}x}{x}
\end{align}
where $\chi$ is the characteristic function.
By Fact \ref{fact:zero}, 
\begin{align}\#\{w_k:\, r<|w_k|\leq x\}=\frac{1}{2}\big[ N_n(\frac{\log x}{2\pi})-N_n(\frac{\log r}{2\pi})\big].
\end{align}
Plugging this into \eqref{eq:int_1} yields
\begin{align}\label{eq:int_1'}
&\sum_{|w_k|>r}\frac{\log r+\log \tilde{R}}{2\log \tilde{R}} \log \frac{|w_k|}{\tilde{R}}+\sum_{|w_k|<1/r}\frac{\log r-\log \tilde{R}}{2\log \tilde{R}} \log |w_k \tilde{R}|\notag\\
=&-\frac{1}{2}\int_r^{\tilde{R}} (N_n(\frac{\log x}{2\pi})-N_n(\frac{\log r}{2\pi}))\, \frac{\mathrm{d}x}{x}\notag\\
=&-\pi\int_{\frac{\log r}{2\pi}}^{\frac{\log \tilde{R}}{2\pi}} (N_n(\varepsilon)-N_n(\frac{\log r}{2\pi}))\, \mathrm{d}\varepsilon.
\end{align}
Similarly,
\begin{align}\label{eq:int_2}
&\sum_{1/r\leq |w_k|\leq r} \frac{\log r-\log \tilde{R}}{2\log \tilde{R}} \log |w_k \tilde{R}|\notag\\
=&\sum_{1<|w_k|\leq r}  \frac{\log r-\log \tilde{R}}{2\log \tilde{R}} \log |w_k \tilde{R}|
+\sum_{1/r\leq |w_k|<1}  \frac{\log r-\log \tilde{R}}{2\log \tilde{R}} \log |w_k \tilde{R}|\notag\\
&+\sum_{|w_k|=1} \frac{\log r-\log \tilde{R}}{2\log \tilde{R}} \log \tilde{R}\notag\\
=&\sum_{1<|w_k|\leq r}  (\log r-\log \tilde{R})+\frac{N_n(0)}{2} (\log r-\log \tilde{R})\notag\\
=&\frac{1}{2} (\log r-\log \tilde{R})N_n(\frac{\log r}{2\pi}).
\end{align}
Therefore, combing \eqref{eq:int_1'} and \eqref{eq:int_2} yields
\begin{align}\label{eq:int_final}
I_n(\log r, \log \tilde{R})=-\frac{\pi}{n} \int_{\frac{\log r}{2\pi}}^{\frac{\log \tilde{R}}{2\pi}} N_n(\varepsilon)\, \mathrm{d}\varepsilon.
\end{align}
Integrating \eqref{eq:thm_zeros} along $z\in \mathcal{C}_{r_j}$, $1\leq r_1<r_2\leq \tilde{R}$, and combining with \eqref{eq:int_final}, one obtains 
\begin{align*}
\int_0^{1} u_n(r_j e^{2\pi i\theta})\, \mathrm{d}\theta=-\frac{\pi}{n}\int_{\frac{\log r_j}{2\pi}}^{\frac{\log \tilde{R}}{2\pi}} N_n(\varepsilon)\, \mathrm{d}\varepsilon+\int_0^{1} h_{\tilde{R},n}(r_je^{2\pi i\theta})\, \mathrm{d}\theta.
\end{align*}
Taking the difference of the equations above between $r_1$ and $r_2$, we arrive at 
\begin{align}\label{eq:N_0'}
\int_0^{1} u_n(r_2 e^{2\pi i\theta})\, \mathrm{d}\theta-\int_0^{1} u_n(r_1e^{2\pi i\theta})\, \mathrm{d}\theta=&\frac{\pi}{n}\int_{\frac{\log r_1}{2\pi}}^{\frac{\log r_2}{2\pi}} N_n(\varepsilon)\, \mathrm{d}\varepsilon \notag\\
&+\int_0^1 h_{\tilde{R},n}(r_2 e^{2\pi i \theta})\, \mathrm{d}\theta-\int_0^1 h_{\tilde{R},n}(r_1 e^{2\pi i \theta})\, \mathrm{d}\theta.
\end{align}
By Lemma \ref{lem:un_lower}, we have for $n$ large,
\begin{align}\label{eq:N_1'}
\int_0^{1} u_n(r_j e^{2\pi i\theta})\, d\theta\geq L_n(\frac{\log r_j}{2\pi})-\frac{1}{n^{\gamma_1}}.
\end{align}
while it follows from  Lemma \ref{lem:vn_upper} and \eqref{eq:un<vn} that for $C=C(\tau,E,f,\eta,c,a)>0$ and $C_0=C_0(a)$,
\begin{align}\label{eq:N_2'}
u_n(r_j e^{2\pi i\theta})\leq v_n(r_j e^{2\pi i\theta})\leq L_n(\frac{\log r_j}{2\pi})+C\frac{(\log n)^{C_0}}{n}.
\end{align}
By Lemma \ref{lem:h_un}, one has for $n$ large,
\begin{align}\label{eq:N_2''}
\left|\int_0^1 h_{\tilde{R},n}(r_2 e^{2\pi i \theta})\, \mathrm{d}\theta-\int_0^1 h_{\tilde{R},n}(r_1 e^{2\pi i \theta})\, \mathrm{d}\theta\right|\leq \frac{C}{\tilde{R}-r_2} n^{-\gamma_1}.
\end{align}
Hence plugging the estimates  \eqref{eq:N_1'} and \eqref{eq:N_2'} into \eqref{eq:N_0'}, one concludes that  for $n$ large,
\begin{align}\label{eq:N_3'}
\frac{1}{2n}(\log r_2-\log r_1) N_n(\frac{\log r_1}{2\pi})\leq &\frac{\pi}{n}\int_{\frac{\log r_1}{2\pi}}^{\frac{\log r_2}{2\pi}} N_n(\varepsilon)\, \mathrm{d}\varepsilon\notag\\
\leq &L_n(\frac{\log r_2}{2\pi})-L_n(\frac{\log r_1}{2\pi})+\frac{C}{\tilde{R}-r_2} n^{-\gamma_1},
\end{align}
and
\begin{align}\label{eq:N_4'}
\frac{1}{2n}(\log r_2-\log r_1) N_n(\frac{\log r_2}{2\pi})\geq &\frac{\pi}{n}\int_{\frac{\log r_1}{2\pi}}^{\frac{\log r_2}{2\pi}} N_n(\varepsilon)\, \mathrm{d}\varepsilon \notag\\
\geq &L_n(\frac{\log r_2}{2\pi})-L_n(\frac{\log r_1}{2\pi})-\frac{C}{\tilde{R}-r_2}n^{-\gamma_1}.
\end{align}
Taking $r_1=\exp(2\pi(\varepsilon_0+\frac{1}{3}\varepsilon_1))$ and $r_2=\exp(2\pi(\varepsilon_0+\frac{2\varepsilon_1}{3}))$ in \eqref{eq:N_3'} yields
\begin{align}\label{eq:N_5'}
\frac{\pi \varepsilon_1}{3n} N_n(\varepsilon_0+\frac{1}{3}\varepsilon_1)\leq L_n(\varepsilon_0+\frac{2}{3}\varepsilon_1)-L_n(\varepsilon_0+\frac{1}{3}\varepsilon_1)+C \varepsilon_1^{-1} n^{-\gamma_1}.
\end{align}
Setting $r_1=\exp(2\pi\varepsilon_0)$ and $r_2=\exp(2\pi(\varepsilon_0+\frac{\varepsilon_1}{3}))$ in \eqref{eq:N_4'} yields
\begin{align}\label{eq:N_6'}
\frac{\pi\varepsilon_1}{3n} N_n(\varepsilon_0+\frac{1}{3}\varepsilon_1)\geq L_n(\varepsilon_0+\frac{\varepsilon_1}{3})-L_n(\varepsilon_0)-C\varepsilon_1^{-1} n^{-\gamma_1}.
\end{align}
Combining \eqref{eq:N_5'}, \eqref{eq:N_6'} with \eqref{eq:zero_count_LE}, we infer that for some absolute constant $C_1>0$ and $n>N(\tau,E,f,\eta,\alpha,\gamma_1)$,
\begin{align}\label{eq:N_7'}
\left| \frac{1}{2n} N_n(\varepsilon_0+\frac{1}{3}\varepsilon_1)-\kappa\right| \leq C_1\varepsilon_1^{-2} n^{-\gamma_1}.
\end{align}
This proves the claimed result.
\end{proof}

\subsection{Riesz representation for $u_n$ via zeros}

We can now draw the following conclusions about the structure of the logarithms of $|D_n|$ where $D_n$ are  the determinants in finite volume. 

\begin{theorem}\label{thm:Riesz_un}
Let $E\in \R$ be such that $L(E,0)\geq \tau>0$. 
Suppose  $\varepsilon_2\in (0,\eta)$ satisfies 
\begin{align}\label{eq:L_linear}
L(E,\varepsilon)=L(E,0)+2\pi \kappa(E,0)\varepsilon,
\end{align}
for all $0\leq \varepsilon\leq \varepsilon_2$, and $D_n(z.E)$ is zero-free on $\partial A_{e^{2\pi \varepsilon_2}}$.
Let $R=e^{2\pi\varepsilon_2}$ and $w_1,...,w_{N_n(E,\varepsilon_2)}$ be the zeros of $D_n(z,E)$ in $A_{R}$
and  define 
\begin{align*}
G_{R,n}(z,E)=\frac{1}{n}\sum_{k=1}^{N_n(E,\varepsilon_2)} G_R(z, w_k),
\end{align*}
where $G_R$ is the Green's function in \eqref{def:G}. Then 
\begin{align*}
u_n(z,E)=2\pi G_{R,n}(z,E)+h_{R,n}(z,E),
\end{align*}
where the harmonic part satisfies $h_{R,n}=u_n$ on $\partial A_R$.
Furthermore, 
\begin{itemize}
\item with the constant $\gamma_1$ in Lemma \ref{lem:un_lower} and the constant $C_0$ in Lemma~\ref{lem:vn_upper}, for any $z\in A_{r}$, $1\leq r<R$, and $n>N(\tau,E,f,\eta,\alpha)$ the harmonic part satisfies
\begin{align}\label{eq:h_constant_un*}
L_n(E,\varepsilon_2)-\frac{C}{R-r} \frac{1}{n^{\gamma_1}}\leq  h_{R,n}(z,E)\leq L_n(E,\varepsilon_2)+C\frac{(\log n)^{C_0}}{n},
\end{align}
for some constant $C=C(\tau,E,f,\eta,\alpha)>0$.
\item with the constant $\gamma_1$ in Lemma \ref{lem:un_lower} and some absolute constant $C>0$,   for all $0\leq \varepsilon\leq \frac{2}{3}\varepsilon_2$, and $n>N(\tau,E,f,\eta,\alpha,\gamma_1)$ 
\begin{align}\label{eq:N_est_un}
|\frac{1}{2n} N_n(E, \varepsilon )-\kappa(E,0)|\leq C\varepsilon^{-2} n^{-\gamma_1}.
\end{align}
In particular, when the acceleration $\kappa(E,0)=1$, 
\begin{align}\label{eq:N_est_un_kappa=1}
|\frac{1}{2n} N_n(E, \varepsilon)-1|\leq C\varepsilon^{-2} n^{-\gamma_1}.
\end{align}
\end{itemize}
\end{theorem}
\begin{proof}
Taking $\varepsilon_0=\varepsilon_1=\frac{3}{4}\varepsilon\leq \frac{1}{2}\varepsilon_2$, it is clear that \eqref{eq:N_est_un} follows from Theorem \ref{thm:zero_count}.
Taking $\varepsilon_0=\varepsilon_1=\frac{1}{2}\varepsilon_2$, \eqref{eq:h_constant_un} follows from Lemma \ref{lem:h_un}.
\end{proof}

\section{Riesz representation for $v_n$}\label{sec:Riesz_vn}
It is natural to ask if the acceleration can be characterized by the Riesz mass of the function $v_n$ rather than through the number of zeros of $D_n$, or equivalently,  the Riesz mass of $u_n$. For future reference, we show here that this is indeed the case.

\begin{theorem}\label{thm:Riesz_vn}
Let $E\in \R$ be such that $L(E,0)\geq \tau>0$. 
Let $\varepsilon_2\in (0,\eta)$ be such that
\begin{align}\label{eq:L_linear_vn}
L(E,\varepsilon)=L(E,0)+2\pi \kappa(E,0)\varepsilon,
\end{align}
for $0\leq \varepsilon\leq \varepsilon_2$ and set $R=e^{2\pi\varepsilon_2}$.
Then 
\begin{align}\label{eq:Riesz_vn}
v_n(w,E)=\int_{A_R} 2\pi G_R(z,w)\, \mu_{n,E}(\mathrm{d}z)+h_{R,n}(w,E),
\end{align}
where the harmonic part satisfies $h_{R,n}=v_n$ on $\partial A_R$.
Furthermore, for any $1< r<R^{1/3}$,  for some constant $C>0$, and $n$ large,
\begin{align}\label{eq:Riesz_mass}
|\mu_{n,E}(A_r)-2\kappa(E,0)|\leq \frac{C \varepsilon_2^{-1} }{\log r} \cdot \frac{(\log n)^{C_0}}{n}.
\end{align}
With the constant $C_0$ in Lemma \ref{lem:vn_upper} and some constant $C>0$,  for any $1\leq r<R$ and $n$ large enough
\begin{align}\label{eq:harmonic_vn}
L_n(E,\varepsilon_2)-\frac{C}{R-r}\frac{(\log n)^{C_0}}{n}\leq h_{R,n}(w,E)\leq L_n(E,\varepsilon_2)+C\frac{(\log n)^{C_0}}{n}
\end{align}
for all  $w\in A_r$.
\end{theorem}
\begin{proof}
We omit the dependence on $E$ in the proof for simplicity.
The analysis of the harmonic part \eqref{eq:harmonic_vn} is similar to that of \eqref{eq:h_constant_un}, and we leave the details to the reader.
Integrating \eqref{eq:Riesz_vn} along $w\in \mathcal{C}_r$ and $C_{r^2}$, $1<r\leq R^{1/3}$, and subtracting one from the other yields in analogy to~\eqref{eq:N_0'} that 
\begin{align}\label{eq:RM_0}
& L_n(\frac{2\log r}{2\pi})-L_n(\frac{\log r}{2\pi})=\int_0^{1} v_n(r^2 e^{2\pi i\theta})\, \mathrm{d}\theta-\int_0^{1} v_n(r e^{2\pi i\theta})\, \mathrm{d}\theta\\
=&\int_{A_R} \left(\int_0^{1}2\pi G_R(z,r^2 e^{2\pi i\theta})\, \mathrm{d}\theta-\int_0^{1} 2\pi G_R(z,re^{i\theta})\, \mathrm{d}\theta\right)\, \mu_{n,E}(\mathrm{d}z)\\
&+\int_0^1 h_{R,n}(r^2 e^{2\pi i\theta})\, \mathrm{d}\theta-\int_0^1 h_{R,n}(re^{2\pi i\theta})\, \mathrm{d}\theta.
\end{align}
By Theorem \ref{thm:Ln_L} and \eqref{eq:L_linear_vn}, for large $n$
\begin{align}\label{eq:RM_0'}
|L_n(\frac{2\log r}{2\pi})-L_n(\frac{\log r}{2\pi})-\kappa(0) \log r|\leq \frac{C}{n}.
\end{align}
By \eqref{eq:harmonic_vn}, 
\begin{align}\label{eq:RM_0''}
\left| \int_0^1 h_{R,n}(r^2 e^{2\pi i\theta})\, \mathrm{d}\theta-\int_0^1 h_{R,n}(re^{2\pi i\theta})\, \mathrm{d}\theta\right|\leq \frac{C}{R-r^2} \frac{(\log n)^{C_0}}{n}.
\end{align}
By Lemma \ref{lem:int_Green} and the symmetry $G_R(z,w)=G_R(w,z)$, 
\begin{align}\label{eq:RM_1}
&\int_0^{1}2\pi G_R(z,re^{2\pi i\theta})\, \mathrm{d}\theta-\int_0^{1}2\pi G_R(z,e^{2\pi i\theta})\, \mathrm{d}\theta\\
=&\begin{cases}
\frac{\log r}{2\log R} \log \frac{|z|}{R}\, \text{ if } |z|\geq r^2\\
\\
\frac{\log r-2\log R}{2\log R}\log |z|+\frac{3}{2}\log r,\, \text{ if } r\leq |z|<r^2\\
\\
\frac{\log r}{2\log R} \log (|z|R), \text{ if } |z|<r
\end{cases}
\end{align}
Since $v_n(z)=v_n(1/\overline{z})$, and $\Delta v_n=\mu_n$, the measure $\mu_n$ exhibits reflection symmetry
$$\mu_n(\mathrm{d}z)=\mu_n(\mathrm{d}(1/\overline{z})).$$
In combination with \eqref{eq:RM_1} we conclude that
\begin{align}\label{eq:RM_2}
\int_{|z|\geq r^2}\frac{\log r}{2\log R} \log \frac{|z|}{R}\, \mu_n(\mathrm{d}z)+\int_{|z|\leq r^{-2}}\frac{\log r}{2\log R}\log (|z|R)\, \mu_n(\mathrm{d}z)=0,
\end{align}
as well as
\begin{align}\label{eq:RM_3}
&\int_{r\leq |z|<r^2} \left(\frac{\log r-2\log R}{2\log R}\log |z|+\frac{3}{2}\log r\right)\, \mu_n(\mathrm{d}z) \notag\\
&+\int_{r^{-2}<|z|\leq r^{-1}} \frac{\log r}{2\log R} \log (|z|R)\, \mu_n(\mathrm{d}z) \notag\\
=&\int_{r\leq |z|<r^2} \left(\frac{\log r-2\log R}{2\log R}\log |z|+\frac{3}{2}\log r+\frac{\log r}{2\log R} \log \frac{R}{|z|}\right)\, \mu_n(\mathrm{d}z) \notag\\
=&\int_{r\leq |z|<r^2} \log \frac{r^2}{|z|}\, \mu_n(\mathrm{d}z),
\end{align}
and
\begin{align}\label{eq:RM_4}
\int_{1< |z|<r} \frac{\log r}{2\log R} \log (|z|R)\, \mu_n(\mathrm{d}z)&+\int_{r^{-1}<|z|< 1}  \frac{\log r}{2\log R} \log (|z|R)\, \mu_n(\mathrm{d}z) 
+\int_{|z|=1} \frac{\log r}{2}\, \mu_n(\mathrm{d}z)\notag\\
=&\frac{\log r}{2}\cdot \mu_n(A_r).
\end{align}
Combining \eqref{eq:RM_0'}, \eqref{eq:RM_0''}, \eqref{eq:RM_2}, \eqref{eq:RM_3} and~\eqref{eq:RM_4} with~\eqref{eq:RM_0}, one obtains
\begin{align}\label{eq:RM_5}
\kappa(0) \log r-\frac{C}{R-r^2}\frac{(\log n)^{C_0}}{n}
\leq &\frac{\log r}{2} \cdot \mu_n(A_r)+\int_{r\leq |z|<r^2} \log \frac{r^2}{|z|}\, \mu_n(\mathrm{d}z)\\
\leq &\frac{\log r}{2} \cdot \mu_n(A_r)+\log r\cdot \mu_n(r\leq |z|<r^2)\\
\leq &\frac{\log r}{2} \cdot \mu_n(A_{r^2}),
\end{align}
as well as
\begin{align}\label{eq:RM_6}
\kappa(0) \log r+\frac{C}{R-r^2}\frac{(\log n)^{C_0}}{n}
\geq &\frac{\log r}{2} \cdot \mu_n(A_r)+\int_{r\leq |z|<r^2} \log \frac{r^2}{|z|}\, \mu_n(\mathrm{d}z)\notag\\
\geq &\frac{\log r}{2} \mu_n(A_r).
\end{align}
In view of \eqref{eq:RM_5} (replacing $r^2$ with $r$) and \eqref{eq:RM_6} one has
\begin{align*}
|\mu_n(A_r)-2\kappa(0)|\leq \frac{C \varepsilon_2^{-1}}{\log r}\cdot  \frac{(\log n)^{C_0}}{n}.
\end{align*}
This proves the claimed result.
\end{proof}

\section{Anderson localization}\label{sec:AL}
In this section, we prove Theorem \ref{thm:loc}. 
Let $E\in \mathcal{S}_2^+$. 
Then $L(E,0)=\tau>0$ and there exists $\varepsilon_2\in (0,\eta)$ such that
\begin{align}\label{eq:L_linear_AL}
L(E,\varepsilon)=L(E,0)+2\pi \varepsilon,
\end{align}
for $0\leq \varepsilon\leq \varepsilon_2$ and $D_n(z,E)$ is zero-free on $\partial A_{e^{2\pi\varepsilon_2}}$.
Let $R=e^{2\pi \varepsilon_2}$.
Note that since the potential $f$ is assumed to be even, $f(\theta)=f(-\theta)$, and thus $D_n(\theta-\frac{n-1}{2}\alpha,E)=D_n(-\theta-\frac{n-1}{2}\alpha,E)$ as well as 
\begin{align}\label{eq:Dn_even}
D_n(z \cdot e^{-\pi i (n-1)\alpha},E)=D_n(e^{-\pi i (n-1)\alpha} /z,E), \text{ for } z\in A_R.
\end{align}
We have by \eqref{eq:Dn_even} the following
\begin{fact}\label{fact:even}
If $w\in A_R$ is a zero of $D_n(z,E)$, then $e^{-2\pi i(n-1)\alpha} /w$ is also a zero.
\end{fact}
In the following, we shall fix an energy $E$ and omit the dependence on $E$ for simplicity.
We shall also write $L_n(E,0)=L_n$ and $L(E,0)=L$.

\subsection{Geometric structure of the large deviation set}
Let $\varepsilon$ be a small constant such that $0<\varepsilon<\min(L(E,0)/20, \varepsilon_2)$, and we set $R':=e^{2\pi \varepsilon}$ and $N':=N_n(\varepsilon)$ for simplicity. 
Let $w_1,...,w_{N'}$ be the zeros of $D_n(z)$ in $A_{R'}$.
By the Riesz representation theorem, see Theorem \ref{thm:Riesz_un},  applied to $u_n$ with $R'$ instead of $R$,  
\begin{align}
u_n(z)=G_{R',n}(z)+h_{R',n}(z).
\end{align}
We need to control the complexity of the large deviation set
\begin{align}\label{def:Bn}
\mathcal{B}_n:=\{\theta\in \T:\, u_n(e^{2\pi i\theta})<L_n-n^{-\gamma_2}\},
\end{align}
where $\gamma_2>0$ is the constant in Lemma \ref{lem:un_LDT}. We are dropping $E$ from some of the notation for simplicity.

\begin{lemma}\label{lem:interval}
For $n$ large enough, there exists an integer $N''\leq \frac{N'}{2}+1$, and a collection $\mathcal{F}_n=\{U_j\}_{j=1}^{N''}$ intervals in $\T$ such that 
\begin{align*}
\mathcal{B}_n \subset \bigcup_{j=1}^{N''} (U_j \cup (-(n-1)\alpha -U_j)),
\end{align*}
where the notation $x-U:=\{x-\theta:\, \theta\in U\}$ for $x\in \T$ and $U\subseteq \T$.
Furthermore, for each $1\leq j\leq N''$,
\begin{align}\label{eq:sum_Uj}
|U_j|\leq e^{-n^{\gamma_2/2}},
\end{align}
where $\gamma_2>0$ is the constant in Lemma \ref{lem:un_LDT}.
\end{lemma}
\begin{proof}
This essentially follows from combining Theorem \ref{thm:acc=zeros} and \cite[Lemma 2.17]{GS2}.
We first recall  Cartan's estimate~\cite[Theorem 4, Page 79]{Levin} as it appears in~\cite[Lemma 2.15]{GS2}

\begin{definition}[Cartan set]
For an arbitrary subset $\mathcal{P}\subset \mathcal{D}(z_0,1)\subset \C$, where $\mathcal{D}(z_0,1)$ is the disk, we say that $\mathcal{P}\in \mathrm{Car}(H,K)$ if $\mathcal{P}\subset \cup_{k=1}^{k_0} \mathcal{D}(z_k, r_k)$ with $k_0\leq K$,  and
\begin{align}\label{eq:Hr}
\sum_{j} r_j<e^{-H}.
\end{align}
\end{definition}

By Wiener's covering lemma  we can assume that $\mathcal{D}(z_k,r_k)$ are pairwise disjoint,
 at the expense of a factor of~$3$ in~\eqref{eq:Hr}. 

\begin{lemma}\label{lem:Cartan}
Let $\varphi$ be an analytic function defined in a disk $\mathcal{D}:=\mathcal{D}(z_0,1)$. Let 
$M\geq \sup_{z\in \mathcal{D}} \log |\varphi(z)|$, $m\leq \log |\varphi(z_0)|$. Given $H\gg 1$, there exists a set $\mathcal{P}\subset \mathcal{D}$, $\mathcal{P}\in \mathrm{Car}(H,K)$, $K=CH(M-m)$ for some absolute constant $C>0$, such that 
\begin{align}
\log |\varphi(z)|>M-CH(M-m),
\end{align}
for any $z\in \mathcal{D}(z_0, 1/6)\setminus \mathcal{P}$.
\end{lemma}
\begin{proof}
See \cite[Lemma 2.15]{GS2}.
\end{proof}

By Lemma \ref{lem:un_LDT}, for $n$ large enough, we have
\begin{align*}
|\mathcal{B}_n|\leq e^{-n^{\gamma_2}}.
\end{align*}
Hence we can find $\{\theta_j\}_{j=1}^{j_n}$ such that for any $j$, $\theta_j\notin \mathcal{B}_n$, and also $j_n\leq 20n$, $\T=\bigcup_{j} (\theta_j-\frac{1}{12n}, \theta_j+\frac{1}{12n})$, and thus
\begin{align}\label{eq:Car_C1_cover}
\mathcal{C}_1\subset \bigcup_{j} \mathcal{D}(e^{2\pi i \theta_j}, \frac{\pi}{6n}).
\end{align}
Let $z_{0,j}:=e^{2\pi i\theta_j}$. 
Consider $\varphi_j(z):=D_n(\frac{2\pi}{n}(z-z_{0,j})+z_{0,j},E)$ as an analytic function on $\mathcal{D}_j:=\mathcal{D}(z_{0,j},1)$.
Since $\theta_j\notin \mathcal{B}_n$, we have
\begin{align}\label{eq:Car_low}
\log |\varphi_j(z_{0,j})|=n\cdot u_n(e^{2\pi i \theta_j})\geq nL_n-n^{1-\gamma_2}.
\end{align}
Also since for $z\in \mathcal{D}_j$, we have 
$\left| \frac{2\pi}{n}(z-z_{0,j})+z_{0,j}\right| \in [1-\frac{2\pi}{n}, 1+\frac{2\pi}{n}]$, hence
\begin{align}
\frac{2\pi}{n}(z-z_{0,j})+z_{0,j}\in A_{e^{4\pi/n}}.
\end{align}
Hence by \eqref{eq:un<vn}, and Lemma \ref{lem:vn_upper}, we have for $z\in \mathcal{D}_j$,
\begin{align}\label{eq:Car_1}
\log |\varphi_j(z)|\leq \sup_{w\in A_{e^{4\pi/n}}} n\cdot v_n(w)\leq \sup_{|\varepsilon|\leq 2/n} n L_n(\varepsilon)+C(\log n)^{C_0}.
\end{align}
By Theorem \ref{thm:Ln_L} and \eqref{eq:L_linear_AL}, we have
\begin{align}\label{eq:Car_2}
\sup_{|\varepsilon|\leq 2/n} L_n(\varepsilon)\leq L_n+\frac{C}{n}.
\end{align}
Hence for $z\in \mathcal{D}_j$, by \eqref{eq:Car_1} and \eqref{eq:Car_2}, we have
\begin{align}\label{eq:Car_upp}
\log |\varphi_j(z)|\leq nL_n+C(\log n)^{C_0}.
\end{align}
By Lemma \ref{lem:Cartan} with $H=n^{\gamma_2/2}$ and the estimates in \eqref{eq:Car_low} and \eqref{eq:Car_upp}, there exists $\mathcal{P}_j\subset \mathcal{D}_j$, $\mathcal{P}_j\in \mathrm{Car}(H,K)$, $K=CH(M-m)$ such that
\begin{align}\label{eq:Car_3}
\log |\varphi_j(z)|>nL_n+C(\log n)^{C_0}-CH(C(\log n)^{C_0}+n^{1-\gamma_2})>nL_n-C_1 n^{1-\frac{1}{2}\gamma_2},
\end{align}
for some constant $C_1$ and any $z\in \mathcal{D}(z_{0,j}, 1/6)\setminus \mathcal{P}_j$.
Next we show
\begin{lemma}\label{lem:Car_zero}
For any $\theta_0$, such that 
\begin{align}\label{eq:Car_t0}
u_n(e^{2\pi i\theta_0})<L_n-C_1 n^{-\frac{1}{2}\gamma_2},
\end{align}
there exists a zero $w_{\ell}$ of $D_n(z)$, such that 
\begin{align}\label{eq:Car_goal}
|e^{2\pi i\theta_0}-w_{\ell}|\leq \frac{4\pi}{n}e^{-n^{\gamma_2/2}}.
\end{align}
\end{lemma}
\begin{proof}
By \eqref{eq:Car_C1_cover}, there exists $\theta_{j_*}$, $1\leq j_*\leq j_n$, such that 
$$|e^{2\pi i\theta_0}-e^{2\pi i\theta_{j_*}}|\leq \frac{\pi}{6n},$$
which implies
\begin{align}\label{eq:Car_t0_in}
z_{0,j_*}+\frac{n}{2\pi}(e^{2\pi i\theta_0}-z_{0,j_*})\in \mathcal{D}_{j_*}(z_{0,j_*}, \frac{1}{12}).
\end{align}
Then by \eqref{eq:Car_t0} and that 
$$\log |\varphi_{j_*}(z_{0,j_*}+\frac{n}{2\pi}(e^{2\pi i\theta_0}-z_{0,j_*}))|=n\cdot u_n(e^{2\pi i \theta_0}),$$ 
whence
\begin{align}\label{eq:Car_t0'}
\log |\varphi_{j_*}(z_{0,j_*}+\frac{n}{2\pi}(e^{2\pi i\theta_0}-z_{0,j_*}))|<nL_n-C_1 n^{1-\frac{1}{2}\gamma_2}.
\end{align}
By \eqref{eq:Car_3}, it is necessary that 
$$z_{0,j_*}+\frac{n}{2\pi}(e^{2\pi i\theta_0}-z_{0,j_*})\in \mathcal{P}_{j_*}\subset \bigcup_{k=1}^{k_*} \mathcal{D}(z_{k,j_*}, r_{k,j_*}).$$
Let $k_0$ be such that
\begin{align}\label{eq:Car_5}
z_{0,j_*}+\frac{n}{2\pi}(e^{2\pi i\theta_0}-z_{0,j_*})\in \mathcal{D}_{j_*}(z_{0,j_*}, \frac{1}{12})\cap \mathcal{D}(z_{k_0,j_*}, r_{k_0,j_*}).
\end{align}
Then since $r_{k_0,j_*}\ll 1/6$,
\begin{align}
\mathcal{D}(z_{k_0,j_*}, r_{k_0,j_*})\subset \mathcal{D}_{j_*}(z_{0,j_*}, \frac{1}{6}).
\end{align}
Hence by \eqref{eq:Car_3}, we have
\begin{align}\label{eq:Car_4}
\log |\varphi_j(z)|\geq nL_n-C_1 n^{1-\frac{1}{2}\gamma_2}, \text{ for } z\in \partial \mathcal{D}(z_{k_0,j_*}, r_{k_0,j_*}).
\end{align}
Assume $\phi_j(z)\neq 0$ for $z\in \mathcal{D}_{j_*}(z_{0,j_*}, \frac{1}{6})$. Then,  by the maximal principle for harmonic functions, it follows from~\eqref{eq:Car_4} that
\begin{align}
\log |\varphi_j(z)|\geq nL_n-C_1 n^{1-\frac{1}{2}\gamma_2}, \text{ for } z\in \mathcal{D}(z_{k_0,j_*}, r_{k_0,j_*})
\end{align}
But this leads to a contradiction with \eqref{eq:Car_t0_in} and \eqref{eq:Car_t0'}.
Hence 
$$\varphi_j(\tilde{z})=D_n\big(\frac{2\pi}{n}(\tilde{z}-z_{0,j_*})+z_{0,j_*},E\big)=0$$ for some 
$\tilde{z}\in \mathcal{D}(z_{k_0,j_*}, r_{k_0,j_*})$. 
This implies
\begin{align}\label{eq:Car_6}
\mathcal{D}(z_{k_0,j_*}, r_{k_0,j_*})\ni \tilde{z}=z_{0,j*}+\frac{n}{2\pi}(w_{\ell}-z_{0,j_*}), \text{ for some } \ell.
\end{align}
Combining \eqref{eq:Car_5} with \eqref{eq:Car_6}, we have
\begin{align*}
\frac{n}{2\pi}\left| w_{\ell}-e^{2\pi i\theta_0}\right|\leq 2r_{k_0, j_*}\leq 2e^{-H}=2e^{-n^{\gamma_2/2}}.
\end{align*}
This proves \eqref{eq:Car_goal}. 
\end{proof}
By Lemma \ref{lem:Car_zero} and Theorem \ref{thm:acc=zeros}, we have
\begin{align}
\mathcal{B}_n\subset \bigcup_{\ell=1}^{N'} \{\theta\in \T:\, |e^{2\pi i \theta}-w_{\ell}|\leq e^{-n^{\gamma_2/2}}\}.
\end{align}
By Fact \ref{fact:even}, 
\begin{align}
\bigcup_{\ell=1}^{N'} \{\theta\in \T:\, |e^{2\pi i \theta}-w_{\ell}|\leq e^{-n^{\gamma_2/2}}\}=\bigcup_{j=1}^{N''}(U_j\cup (-(n-1)\alpha-U_j)),
\end{align}
for some $N''\leq \frac{N'}{2}+1$, as claimed by Lemma \ref{lem:interval}.
\end{proof}

\subsection{Proof of Anderson localization}
To prove Anderson localization, by Shnol's theorem \cites{Be,Simon,S}, it suffices to show that any generalized eigenfunction $\phi$ with the property that 
\begin{align}\label{eq:Schnol}
\max(|\phi_0|, |\phi_{-1}|)=1, \text{ and } |\phi_y|\leq C|y|,
\end{align}
decays exponentially. 
In the following, let $\phi$ be a solution  of $H_{\alpha,\theta}\phi=E\phi$, satisfying \eqref{eq:Schnol}.
Combining Theorem \ref{thm:Riesz_un} (with $\kappa=1$) with Lemma~\ref{lem:interval} yields 
\begin{align*} 
\mathcal{B}_n\subset \bigcup_{j=1}^{N''} (U_j\cup (-(n-1)\alpha-U_j)),
\end{align*}
with $N''\leq n+C\varepsilon_2^{-1} n^{1-\gamma_1}$. 
Furthermore, for each $j$, one has the measure estimate
\begin{align}\label{eq:length_Ij}
|U_j|\leq e^{-n^{\gamma_2}}.
\end{align}
For $x\in \R$, let $[x]$ be the integer part of $x$.
\begin{lemma}\label{lem:I1_I2}
For any $n$ large enough, and any $y\in \Z$ such that\footnote{The proof for negative $y$ is analogous by symmetry.} $n<y<10n$, 
let 
\begin{align}
I_1:=&[-[\frac{7}{8}n], -[\frac{1}{8}n]]\\
I_2:=&[y-[\frac{7}{8}n], y-[\frac{1}{8}n]]
\end{align}
There exists $\ell\in I_1\cup I_2$ such that 
\begin{align*}
\theta+\ell \alpha\notin \bigcup_{j=1}^{N''} (U_j \cup ((n-1)\alpha-U_j)).
\end{align*}
\end{lemma}
\begin{proof}
Note that the cardinality
\begin{align}
\# I_1+\# I_2\geq \frac{3}{2} n-2\geq N'',
\end{align}
for $n$ large enough.
It suffices to prove each pair $U_j\cup (-(n-1)\alpha-U_j)$ consists of at most one point in  $\{\theta+\ell\alpha\}_{\ell\in I_1\cup I_2}$. 
We argue by contradiction, suppose there exist $\ell_1, \ell_2$ such that
\begin{align*}
\theta+\ell_1\alpha\in U_j, \text{ and } \theta+\ell_2\alpha\in U_j.
\end{align*}
Then, since $\alpha\in \mathrm{DC}_{c,a}$ and that $|\ell_1-\ell_2|<11n$, 
\begin{align*}
|U_j|\geq \|\theta+\ell_1\alpha-(\theta+\ell_2\alpha)\|_{\T}=\|(\ell_1-\ell_2)\alpha\|_{\T}\geq \frac{c''}{n(\log n)^a},
\end{align*}
for some constant $c''$ depending on $c$.
But this contradicts with \eqref{eq:length_Ij}.
The case when 
\begin{align*}
\theta+\ell_1\alpha\in (-(n-1)\alpha-U_j), \text{ and } \theta+\ell_2\alpha\in (-(n-1)\alpha-U_j). 
\end{align*}
is similar.
Suppose there exist $\ell_1, \ell_2$ such that
\begin{align*}
\theta+\ell_1\alpha\in U_j, \text{ and } \theta+\ell_2\alpha\in (-(n-1)\alpha-U_j). 
\end{align*}
Since $\theta\in (\Theta_{c',b})^c$, for $k$ large enough, one has 
\begin{align*}
\|2\theta+k\alpha\|_{\T}\geq \frac{c'}{|k|^b}.
\end{align*}
Using that $n/4\leq \ell_1+\ell_2+n\leq 11n$, we infer that 
\begin{align*}
|U_j|\geq \|\theta+\ell_1\alpha-(-\theta-\ell_2\alpha-(n-1)\alpha)\|_{\T}=\|2\theta+(\ell_1+\ell_2+n-1)\alpha\|_{\T}\geq \frac{\tilde{c}}{n^{b}},
\end{align*}
for some constant $\tilde{c}$ depending on $c'$. This contradicts with \eqref{eq:length_Ij} again. 
Thus the claimed results hold.
\end{proof}

Next, we show the following.
\begin{lemma}\label{lem:I1_small}
For any $\ell\in I_1$, one has $\theta+\ell\alpha\in \bigcup_{j=1}^{N''}(U_j\cup (-(n-1)\alpha-U_j))$.
\end{lemma}
\begin{proof}
Argue by contradiction. 
Suppose there exists $\ell_1\in I_1$ such that 
$$\theta+\ell_1\alpha\notin \bigcup_{j=1}^{N''}(U_j\cup (-(n-1)\alpha-U_j)).$$
By Lemma \ref{lem:interval}, it is necessary that $\theta+\ell_1\alpha\notin \mathcal{B}_n$, which implies
\begin{align}\label{eq:u_I1_large}
\frac{1}{n}\log |D_n(\theta+\ell_1\alpha)|=u_n(e^{2\pi i(\theta+\ell_1\alpha)})\geq L_n-n^{-\gamma_2/2}\geq L-\varepsilon,
\end{align}
where we used that $L_n\geq L$ due to Theorem \ref{thm:Ln_L}.
Let $\ell_2:=\ell_1+n-1$. 
By Lemma \ref{lem:vn_upper}, Theorem \ref{thm:Ln_L} and \eqref{eq:un<vn}, we see that for $k>k(\varepsilon)$ large enough
\begin{align}\label{eq:Dk_upper}
\frac{1}{k}\log |D_k(\theta)|=u_k(e^{i\theta})\leq v_k(e^{i\theta})\leq L_k+C\frac{(\log k)^{C_0}}{k}\leq L+\varepsilon.
\end{align}
Combining \eqref{eq:u_I1_large}, \eqref{eq:Dk_upper} with \eqref{eq:Green_exp},  
\begin{align}\label{eq:phi0}
|\phi_0|
\leq &\frac{|D_{\ell_2}(\theta+\alpha)|}{|D_n(\theta+\ell_1\alpha)|} |\phi_{\ell_1-1}|+\frac{|D_{-\ell_1}(\theta+\ell_1\alpha)|}{|D_n(\theta+\ell_1\alpha)} |\phi_{\ell_2+1}|\notag\\
\leq &e^{\ell_1 (L-20\varepsilon)} |\phi_{\ell_1-1}|+e^{-\ell_2 (L-20\varepsilon)} |\phi_{\ell_2+1}|\notag\\
\leq &C e^{\ell_1(L-20\varepsilon)} |\ell_1|+C e^{-\ell_2 (L-20\varepsilon)} |\ell_2|<\frac{1}{2}, 
\end{align}
invoking \eqref{eq:Schnol} and $\min(|\ell_1|, |\ell_2|)\geq [n/8]$. 
Similarly, one shows that $|\phi_{-1}|<1/2$.
Hence we arrive at a contradiction with the assumption that $\max(|\phi_0|, |\phi_{-1}|)=1$.
\end{proof}

Combining Lemmas \ref{lem:I1_I2}, \ref{lem:I1_small} with Lemma \ref{lem:interval} yields
\begin{corollary}\label{cor:I2_large}
There exists $\ell_3\in I_2$ such that $\theta+\ell_3\alpha\notin \mathcal{B}_n$.
\end{corollary}
The proof of Anderson localization then follows from a similar argument as in the proof of Lemma~\ref{lem:I1_small}.
Indeed, let $\ell_4:=\ell_3+n-1$. Similar to \eqref{eq:phi0}, we have
\begin{align}\label{eq:phiy}
|\phi_y|\leq C e^{-(y-\ell_3)(L-20\varepsilon)} \ell_3+C e^{-(\ell_4-y)(L-20\varepsilon)}\ell_4.
\end{align}
By construction of $I_2$, 
\begin{align*}
\min(y-\ell_3, \ell_4-y)\geq \Big[\frac{1}{8}n\Big]\geq \frac{1}{90}y, \text{ and } \max(\ell_3, \ell_4)\leq y+\frac{7}{8}n<2y.
\end{align*}
Plugging the above estimates into \eqref{eq:phiy} yields
\begin{align*}
|\phi_y|\leq e^{-\frac{1}{100}(L-20\varepsilon)y}.
\end{align*}
This proves the claimed result.
\begin{remark}
It is possible to modify the proof to show the following asymptotics:
\begin{align*}
\lim_{|y|\to \infty} -\frac{\ln (|\phi_y|^2+|\phi_{y-1}|^2)}{2|y|}=-L.
\end{align*}
\end{remark}

\section{H\"older regularity of the IDS}\label{sec:IDS}

We now indicate the modifications needed of \cite[Theorem 1.1]{GS2} to prove Theorem~\ref{thm:IDS}. The main change occurs on page~848, in terms of the zero count of the Determinants $D_N$ in the annulus $A_R$. In fact, the $2k_0N$ estimate of the total number of zeros in terms of the $k_0=\deg V$ (where $V$ is the potential function), is now replaced by the sharper Theorem~\ref{thm:acc=zeros}. Thus, \cite[Corollary~14.14]{GS2} can be improved to read
\[
k(\widehat{A}_{j_0},\zeta_{j_0},r^{(2)}) \le 2 \kappa(E,0) \le 2(\ell-1)
\]
for all energies $E\in \bigcup_{j=1}^{\ell} \mathcal{S}_j^{+}$. By the proof of Theorem~1.4 on page~849 of~\cite{GS2} we conclude that that theorem holds with $k_0\le 2\kappa(E,0)$ for all $E\in \cup_{j=1}^{\ell} \mathcal{S}_j^{+}$. This in turn then improves on \cite[Lemma 17.5]{GS2}, with the same $k_0$. Finally, this improvement of Lemma~17.5 allows for the
H{o}lder exponents stated in Theorem~\ref{thm:IDS}, see  \cite[Section 18]{GS2}.


\begin{thebibliography}{99}

\bibitem[Am]{Amor} Amor, S.H., 2009. {\em H\"older Continuity of the Rotation Number for Quasi-Periodic Co-Cycles in ${SL (2,\mathbb R)} $.} Comm.\ math.\ phys., 2(287), pp.\ 565--588.

\bibitem[Av1]{A_ac} Avila, A., 2008. {\em The absolutely continuous spectrum of the almost Mathieu operator.} arXiv preprint arXiv:0810.2965.

\bibitem[Av2]{A_ar_ac} Avila, A., 2010. {\em Almost reducibility and absolute continuity I.} arXiv preprint arXiv:1006.0704.

\bibitem[Av3]{Global} Avila, A., 2015. {\em Global theory of one-frequency Schr\"odinger operators.} Acta Math., 215 (1), pp.\ 1--54.

\bibitem[AD]{AD} Avila, A. and Damanik, D., 2008. {\em Absolute continuity of the integrated density of states for the almost Mathieu operator with non-critical coupling.} Invent.\ math., 172(2), pp.\ 439--453.


\bibitem[AJ]{AJ} Avila, A. and Jitomirskaya, S., 2009. {\em Almost localization and almost reducibility.} J.\ Eur.\ Math.\ Soc., 12(1), pp.\ 93--131.


\bibitem[AYZ]{AYZ} Avila, A., You, J. and Zhou, Q., 2017. {\em Sharp phase transitions for the almost Mathieu operator.} Duke Math.\ J.,  166(14), pp.\ 2697--2718.

\bibitem[B]{Be}Berezanskiĭ, I.M., 1968. {\em Expansions in eigenfunctions of selfadjoint operators (Vol. 17).} AMS.

\bibitem[BG]{BG} Bourgain, J. and Goldstein, M., 2000. {\em On nonperturbative localization with quasi-periodic potential.} Ann.\ of Math., 152 (3),
 pp.\ 835--879.

\bibitem[DGSV]{DGSV}  Damanik, D.,  Goldstein, M., Schlag, W., Voda, M., 2018. {\em Homogeneity of the spectrum for quasi-periodic
Schr\"odinger operators},  J.\  Eur.\  Math.\ Soc.\ 20, pp.\  3073-3111. 

\bibitem[DS]{DS} Dinaburg, E.I. and Sinai, Y.G., 1975. {\em The one-dimensional Schr\"odinger equation with a quasiperiodic potential.}    Funkcional.\ Anal.\ i Prilozen.\ 9, no.\ 4, 8--21.  

\bibitem[E1]{Eliasson} Eliasson, L.H., 1992. {\em Floquet solutions for the $1$-dimensional quasi-periodic Schr\"odinger equation.}  Comm.\ math.\ phys., 146(3), pp.\ 447--482.

\bibitem[E2]{E97} Eliasson, L.H., 1997. {\em Discrete one-dimensional quasi-periodic Schr\"odinger operators with pure point spectrum.} Acta Math., 179(2), pp.\ 153--196.

\bibitem[FV]{FV} Forman, Y. and VandenBoom, T., 2021. {\em Localization and Cantor spectrum for quasiperiodic discrete Schr\"odinger operators with asymmetric, smooth, cosine-like sampling functions.} arXiv preprint arXiv:2107.05461.

\bibitem[FSW]{FSW} Fr\"ohlich, J., Spencer, T. and Wittwer, P., 1990. {\em Localization for a class of one-dimensional quasi-periodic Schr\"odinger operators.} Comm.\ math.\ phys., 132(1), pp.\ 5--25.

\bibitem[GYZ]{GYZ}Ge, L., You, J. and Zhao, X., 2022. {\em H\"older Regularity of the Integrated Density of States for Quasi-periodic Long-range Operators on $\ell^2 ({\mathbb {Z}}^ d)$}. Communications in Mathematical Physics, 392(2), pp.347-376.

\bibitem[GJZ]{GJZ} Ge, L., Jitomirskaya, S. and Zhao, X., 2022. {\em Stability of the non-critical spectral properties I: arithmetic absolute continuity of the integrated density of states.} arXiv preprint arXiv:2204.11000.

\bibitem[GS1]{GS1} Goldstein, M. and Schlag, W., 2001. {\em H\"older continuity of the integrated density of states for quasi-periodic Schr\"odinger equations and averages of shifts of subharmonic functions.} Ann.\ of Math., pp.\ 155--203.

\bibitem[GS2]{GS2} Goldstein, M. and Schlag, W., 2008. {\em Fine properties of the integrated density of states and a quantitative separation property of the Dirichlet eigenvalues.} Geom.\ Funct.\ Anal., 18 (3), pp.\ 755--869.

\bibitem[HZ]{HZ} Han, R. and Zhang, S., 2022. {\em Large deviation estimates and H\"older regularity of the Lyapunov exponents for quasi-periodic Schr\"odinger cocycles.} Int.\ Math.\ Res.\ Not.\ IMRN 2022(3), pp.\ 1666--1713.

\bibitem[J1]{J94} Jitomirskaya, S.Y., 1994. {\em Anderson localization for the almost Mathieu equation: a nonperturbative proof.} Comm.\ math.\ phys., 165(1), pp.\ 49--57.

\bibitem[J2]{J99} Jitomirskaya, S.Y., 1999. {\em Metal-insulator transition for the almost Mathieu operator.} Ann.\ of Math., pp.\ 1159--1175.

\bibitem[JL]{JL} Jitomirskaya, S. and Liu, W., 2018. {\em Universal hierarchical structure of quasiperiodic eigenfunctions.}   Ann.\ of Math., 187 (3), 
pp.\ 721--776.

\bibitem[L]{Levin} Levin, B.Y., 1996. {\em Lectures on entire functions (Vol. 150).} AMS.

\bibitem[P]{Puig} Puig, J., 2004. {\em Cantor spectrum for the almost Mathieu operator.} Comm.\ math.\ phys., 244(2), pp.\ 297--309.

\bibitem[Sch]{S} Schlag, W., 2022. {\em An introduction to multiscale techniques in the theory of Anderson localization, Part I.} Nonlinear Anal.\ 220, Paper No.\ 112869, 55 pp.

\bibitem[Sim]{Simon}Simon, B., 1982. {\em Schr\"odinger semigroups.} Bulletin of the American Mathematical Society, 7(3), pp.\ 447--526.

\bibitem[Sin]{Sinai} Sinai, Y.G., 1987. {\em Anderson localization for one-dimensional difference Schr\"odinger operator with quasiperiodic potential.} Journal stat.\  phys., 46 (5-6), pp.\ 861--909.

\bibitem[Y]{You} You, J., 2018. {\em Quantitative almost reducibility and its applications.} Proceedings of the International Congress of Mathematicians 2018 (pp.\  2113--2135).

\bibitem[YZ]{YZ} You, J. and Zhang, S., 2014. {\em H\"older continuity of the Lyapunov exponent for analytic quasiperiodic Schr\"odinger cocycle with weak Liouville frequency.} Ergodic Theory Dynam.\ Systems 34(4), pp.\ 1395--1408.


\end{thebibliography}
\end{document}